\documentclass[a4paper,UKenglish,cleveref, autoref, thm-restate]{lipics-v2021}
\usepackage[T1]{fontenc}
\usepackage{xspace}
\usepackage{hyperref}
\usepackage{makecell}
\usepackage{todonotes}
 \newcommand{\OP}{Orienteering Problem\xspace}
\newcommand{\tw}{tree-width\xspace}
\newcommand{\T}{\mathcal{T}}
\newcommand{\OPT}{\text{OPT}}

\newcommand{\eps}{\varepsilon}
\renewcommand{\P}{\mathcal{P}}

\newcommand{\iPi}{\mathit{\Pi}}

\newcommand{\mart}[1]{{#1}}

\usepackage{float}



\nolinenumbers
\hideLIPIcs

\usepackage{orcidlink}

\title{Orienteering (with Time Windows) on Restricted Graph Classes} 

\author{Kevin Buchin}{TU Dortmund, Germany}{kevin.buchin@tu-dortmund.de}{https://orcid.org/0000-0002-3022-7877}{}
\author{Mart Hagedoorn}{TU Dortmund, Germany}{mart.hagedoorn@tu-dortmund.de}{https://orcid.org/0000-0002-8591-3380}{}
\author{Guangping Li}{TU Dortmund, Germany}{guangping.li@tu-dortmund.de}{https://orcid.org/0000-0002-7966-076X}{}
\author{Carolin Rehs}{TU Dortmund, Germany}{carolin.rehs@tu-dortmund.de}{https://orcid.org/0000-0002-8788-1028}{}

\authorrunning{K.~Buchin, M.~Hagedoorn, G.~Li, C.~Rehs} 

\begin{document}
\maketitle

\begin{abstract}

Given a graph with edge costs and vertex profits and given a budget $B$, the Orienteering Problem asks for a walk of cost at most $B$ of maximum profit. Additionally, each profit may be given with a time window within it can be collected by the walk.
%
While the Orienteering Problem and thus the version with time windows are NP-hard in general, it remains open on numerous special graph classes. Since in several applications, especially for planning a route from A to B with waypoints, the input graph can be restricted to tree-like or path-like structures, in this paper we consider orienteering on these graph classes. 

While the Orienteering Problem with time windows is NP-hard even on undirected paths and cycles, and remains so even if all profits must be collected, we show that for directed paths it can be solved in $\mathcal{O}(m \log m)$ time (where $m$ is the total number of time windows), even if each profit can be collected in one of several time windows. The same case is shown to be NP-hard for directed cycles. 

Particularly interesting is the Orienteering Problem on a directed cycle with one time window per profit.
We give an efficient algorithm for the case where all time windows are shorter than the length of the cycle, resulting in a 2-approximation for the general setting. 
Based on the algorithm for directed paths, we further develop a polynomial-time approximation scheme for this problem. For the case where all profits must be collected, we present an $\mathcal{O}(n^4)$-time algorithm. 

For the Orienteering Problem with time windows for the edges, we give a quadratic time algorithm for undirected paths and observe that the problem is NP-hard for trees.

In the variant without time windows, we show that on trees and thus on graphs with bounded \tw the Orienteering Problem remains NP-hard. We present, however, an FPT algorithm to solve orienteering with unit profits that we then use to obtain a ($1+\eps$)-approximation algorithm on graphs with arbitrary profits and bounded \tw, which improves current results on general graphs.





\end{abstract}

\section{Introduction}
The \emph{Orienteering Problem
} has a wide range of applications. The problem originates from the sport of orienteering: Competitors navigate between control points marked on a map, aiming to visit as many points as possible within a given time limit. With the addition of different profits for each control point, this is a frequently applied routing problem, for instance in logistics~\cite{golden, logistics/tsiligirides}, tourism~\cite{AAI/tourist/Souffriau, vansteenwegen2007mobile}, and journey planning~\cite{tour4me, lu2015arc}.
In such applications, the class of graphs to be considered can often be restricted. For instance, when planning a tour from A to B, one can assume that the tour stays close or even largely makes use of the main route between A and B, while taking small detours to visit interesting sites close to this main route. Thus, we can limit the problem to a path-like or tree-like subgraph of the travel network. This scenario is modelled well by graphs with bounded tree-width.

Formally, in the \OP (OP) we are given a graph $G$ with $n = |V(G)|$, a starting point $s \in V(G)$, a budget $B \in \mathbb{R}_{> 0}$, a cost function $c : E(G) \rightarrow \mathbb{R}_{\geq 0}$, and a profit function $\pi : V(G) \rightarrow \mathbb{R}_{\geq 0}$. 
The goal is to find a walk $P = (e_1, \cdots, e_\ell)$, starting at $s$ such that $\sum_{e \in P} c(e) \leq B$, which maximises the total profit, where the profit of each vertex is collected only once, i.e.~$\sum_{v \in \bigcup{\{v\mart{, w} | (v, w) \in P\}}} \pi(v)$. 

The \OP has been a popular topic of research in the last decades (see~\cite{vansteenwegen2011} for a survey). Considerable effort has been devoted to designing exact (exponential-time) algorithms~\cite{feillet2005traveling, fischetti1998solving, LAPORTE1990193, ramesh1992optimal} and practical heuristics~\cite{archetti2014chapter, STAVROPOULOU2019340}. Since --as a generalization of the travelling salesperson problem-- the \OP is NP-hard~\cite{golden}, theoretical work on the \OP has mostly focused on approximation algorithms.
The current best approximation algorithm is a ($2 + \eps$)-approximation given by Checkuri et al.~\cite{chekuri2012improved}.
However, the best lower bound on the approximation ratio is at most 1481/1480 as given by Blum et al.~\cite{blum2007approximation}. 
Furthermore, when the sites lie in a fixed-dimensional Euclidean space and the underlying graph is the complete Euclidean graph, a $(1 + \eps)$-approximation algorithm has been proposed by Chen et al.~\cite{cheneuclidean}. 

\medskip

A natural extension of the \OP are versions with \emph{time windows}, i.e., intervals in time.
In applications, sites often need to be reached in specific time windows due to the availability of certain services, hours of operation, or deadlines. 

In the \OP with time windows, additionally each profit is associated with one or more time windows, and it can only be collected if the walk visits the corresponding vertex during one of these windows.
We distinguish between the version, in which every vertex has one single time window (OP-1TW) and the version with multiple time windows per vertex (OP-MTW).


A restricted version of the \OP with time windows, which is closely related to the Travelling Salesperson Problem with time windows, asks if there is a walk $P$ which collects the profit of every vertex. We denote these problems as the Covering Orienteering Problem with single/multiple time windows (COP-1TW / COP-MTW) and show several hardness results even for these restricted problems.  

Problems closely related to the \OP with time windows have been studied before. 
Most frequently
 considered is the Travelling Salesperson Problem with time windows~\cite{DBLP:journals/mpc/AbeledoFPU13, DBLP:conf/stoc/BansalBCM04, DBLP:journals/disopt/BigrasGS08}. 
Since NP-hardness for Travelling Salesperson Problem implies NP-hardness for these versions with time windows, those problems are most studied in terms of heuristics (see~\cite{DBLP:journals/eor/GunawanLV16} for a survey), approximation approaches~\cite{DBLP:conf/stoc/BansalBCM04} or for restricted graph classes~\cite{DBLP:journals/networks/Tsitsiklis92}.  
Most notably, the work done by Tsitsiklis~\cite{DBLP:journals/networks/Tsitsiklis92} is strongly related to the COP-1TW, since, in some variations discussed by Tsitsiklis, a walk is allowed to visit the same vertex multiple times.
Furthermore, Garg et al.~\cite{DBLP:conf/soda/GargKK21} introduced the multiple time window model for Orienteering, for which they show APX-hardness for trees.

While approximation algorithms for the \OP with time windows have been considered before, there are no algorithmic results on restricted graph classes. This is particularly surprising since the problem with time windows is interesting from both an application and an algorithmic point of view even on very restricted graph classes: Here, an optimal route might pass an important point at a 'wrong' time and have to come back later.

In this paper, we investigate the \OP with time windows. For an overview of the presented results, see Table \ref{tab:overview}. Even for undirected paths, COP-1TW and OP-1TW are NP-hard, what we can follow from a similar problem considered in \cite{DBLP:journals/networks/Tsitsiklis92}. Furthermore, we show that COP-MTW and thus OP-MTW are NP-hard even on directed cycles. 
Most interestingly this is not true for one single time window per vertex: We give a polynomial-time algorithm for the COP-1TW. 
Moreover, we present a dynamic program for OP-MTW on directed paths. 
While the question of NP-hardness remains open for OP-1TW on directed cycles, we give an $\mathcal{O}(n^2 \log n)$-time algorithm for short time windows (which admits an FPT algorithm with respect to the number of long intervals), a 2-approximation, and a polynomial-time approximation scheme.



\begin{table}[t]
\centering
\renewcommand{\arraystretch}{1.3}
\begin{tabular}{l||c|c|}
\cline{2-3}
 & Single time window & Multiple time windows  \\ \hline\hline
\multicolumn{1}{|l||}{Directed path}                & $\mathcal{O}(n \log n)$, Prop.~\ref{prop:dirOP} & $\mathcal{O}(m \log m)$, Prop.~\ref{prop:dirOP}  \\ \hline
\multicolumn{1}{|l||}{Directed cycle}     &  \makecell{\begin{tabular}{rl} COP-1TW:~~~ &  \hspace{-4mm} $\mathcal{O}(n^4)$, Thm.~\ref{thm:cycTSP} \\ OP-1TW:~~~ & 
 \hspace{-4mm} PTAS, Thm.~\ref{thm:OPTWdircyc-ptas} \\ & \hspace{-4mm} FPT, Cor.~\ref{cor:cycleFPT} \end{tabular}}      & NP-hard, Thm.~\ref{thm:cycMTW}             \\ \hline
 \multicolumn{1}{|l||}{Undirected path}   &    NP-hard  \cite{DBLP:journals/networks/Tsitsiklis92}    &        NP-hard  \cite{DBLP:journals/networks/Tsitsiklis92}             \\ \hline
\end{tabular}\vspace{1mm}
\caption{Results for the \OP with one single or multiple time windows per vertex, where $m$ is the total number of time windows.}
\label{tab:overview}
\end{table}



Another natural setting in applications, which is closely related to having time windows for the profits, is that edges cannot be used at all times, for example due to construction sights, road illumination or seasonal scenery. 
Formally, we obtain the \OP in a setting of dynamic graphs (OP-D for short) by adding one or more intervals to each edge in which that edge is traversable. 
Surprisingly, while the setting seems very similar to the \OP with time windows, we show that it is solvable in quadratic time on an undirected path. However, the \OP on dynamic graphs remains NP-hard on trees with unit profit and unit cost.

We further consider the \OP without time windows. We show that though this problem is NP-hard even on trees (and thus also on graphs with bounded \tw), we can give a $1+(\eps)$-approximation algorithm on graphs with bounded \tw. 
This stands in contrast to the \OP on general graphs that does not admit an efficient $(1+\eps)$-approximation unless P=NP.

\section{The \OP with Time Windows}\label{Sec:OPTW}
In this section, we study the \OP with time windows. 
We distinguish between the version of the problem
using a single time window per profit (OP-1TW) and using multiple time windows (OP-MTW). 
An instance of OP-1TW consists of an instance of the \OP together with an interval $[r_i, d_i]$ as time window for each of the vertices $v_i \in V(G)$. A walk $W$ through $G$ is now defined by a sequence of tuples $(v, t)$, where $v$ is a vertex and $t \in \mathbb{N}$ is a \emph{time step} such that if $(v,t)$ is followed by $(v',t')$ in $W$, then $t'\geq t+ c(vv')$. The profit of a vertex $v_i$ can only be collected once by $W$ and only if $W$ passes $v_i$ at time step $t$ with $r_i \leq t \leq d_i$. 

The OP-MTW generalized OP-1TW by replacing the time interval per vertex by a set $\mathcal{I}_{v_i} = \{[r_{i,1}, d_{i, 1}], \cdots [r_{i,m_i}, d_{i, m_i}]\}$ of $m_i$ time intervals. The profit of vertex $v_i$ can then be collected at any time step $t$ where there is an interval 
$[r_{i,j}, d_{i, j}] \in \mathcal{I}_{v_i}$, 
such that $r_{i,j} \leq t \leq d_{i, j}$. 
Let $d_{max}$ be the latest deadline of all time windows and w.l.o.g.~assume that $d_{max} = B$.

 

As we will see in \Cref{sec:OP}, the \OP in general is NP-hard even on trees, and thus remains NP-hard in the version with time windows. 
However, in the setting of \OP with time windows, the problem is NP-hard even for undirected paths, which follows directly from previous results:

%
%
%
The line-TSP problem as described and proven NP-complete by Tsitsiklis~\cite{DBLP:journals/networks/Tsitsiklis92} is closely related to the OP-1TW on undirected paths even with unit edge cost and unit vertex profit. In line-TSP, each job $j \in J$ has an integer position $x_j$ and time window $[r_j, d_j]$ and travel time between two jobs $j, j' \in J$ is defined as $|x_j - x_{j'}|$. The goal in line-TSP is to find a walk which collects the profit of all the jobs in $J$.
The difference to the OP is that in line-TSP two jobs can have the same position, which is not possible when each vertex has exactly one associated profit and all edges have unit costs. However, it is possible to reduce line-TSP to the OP-1TW on an undirected path with unit edge cost and unit vertex profit by a straightforward reduction.

\begin{remark} \label{rem:NP-hard}
    The OP-1TW and thus the OP-MTW on undirected paths is weakly NP-hard even assuming unit edge cost and vertex profit.
\end{remark}

\begin{proof}
From a line-TSP instance, we can construct an OP-1TW instance on an undirected path $G$ as follows.
Sort the jobs $J$ by their position $x$. 
From this sorted list, for each $1 \leq i \leq n$, add a vertex $v_i$ with time window $[n^2 \cdot r_i, n^2 \cdot d_i + n]$ to $V(G)$.
If $x_{i-1} = x_i$, then a single edge $(v_{i-1}, v_{i})$ is added to $E(G)$.
Otherwise, $v_{i-1}$ and $v_{i}$ are connected with a path of $(x_{i+1} - x{i}) \cdot n^2 - 1$ intermediate vertices, whose profit cannot be collected.
Finally, $\OPT$ is the number of distinct vertices collected by an optimal walk with cost $d_{max}$.
Then, a solution exists for line-TSP if and only if $\OPT = |V(G)|$ in our constructed instance.
\end{proof}




\subsection{Directed Path} \label{sec:DAGstw}



While on an undirected path the problem is NP-hard, the \OP with time windows is tractable in polynomial time if the input graph is a directed path. This holds even for multiple time windows:

\begin{proposition}\label{prop:dirOP}
Given an instance of the OP-MTW with a directed path $G$, arbitrary profits, and arbitrary edge weights, the OP-MTW can be solved in $\mathcal{O}(m \log m)$ time, where $m$ is the total number of time windows.
\end{proposition}

\begin{proof}
    Assume w.l.o.g.~that every vertex has at least one time window, i.e.~$m \geq |V(G)|$. 
    In case there exists a vertex $v$ without a time window, $v$ can simply be removed and the incident edges merged.
    Furthermore, let $V(G) = \{v_1, \cdots, v_n\}$ with $E(G) =\{\{v_i, v_{i+1} \} \mid 1 \leq i < n\}$.
    We specify a dynamic programming algorithm. 
    For $1 \leq i \leq n$ and $t \in [0, d_{max}]$, let $W_{i, t}$ be a walk with maximum profit starting at $v_1$ and ending at vertex $v_i$ at time $t$. 
    Furthermore, let $c_i$ be the cost of the edge between $v_{i-1}$ and $v_i$. 
    
    The key observations of our approach are as follows.
    For any $1 \leq i \leq n$ and $t, t' \in [0, d_{max}]$ such that $t' > t$, the profit of $W_{i, t'}$ is greater than or equal to the profit of $W_{i, t}$. Since, $W_{i, t'}$ can always collect at least the profit of $W_{i, t}$ by following the path of $W_{i, t}$ and waiting in $v_i$ until $t'$.
    Therefore, for any $1 \leq i \leq n$ and $t \in [0, d_{max}]$, if $t$ is in a time window of $v_i$, the profit of $W_{i, t}$ is equal to the sum of the profit of $W_{i-1, t-c_i}$ and $\pi(v)$. Otherwise, $W_{i, t}$ has or has not collected $v_i$. If $v_i$ was not collected, the profit of $W_{i, t}$ is simply equal to that of $W_{i-1, t - c_i}$. If, however, $v_i$ was collected by $W_{i, t}$, then the walks profit is equal to the sum of $\pi(v_i)$ and profit of $W_{i, d_{i, j}}$, where $d_{i, j}$ is the latest deadline of $v_i$ before time $t$.
    
    In order to find the walk $W_{i, d_{max}}$ with the highest profit.
    We will maintain a 1D range tree $T$, where the nodes of $T$ correspond to time intervals. 
    Specifically, at iteration $i$, a leaf $\ell$ of $T$ corresponds to interval $[\ell_l, \ell_r)$, and every leaf is augmented with a walk $W_{i, l}$ and a profit value $\pi_\ell$. 
    Each internal node $w \in T$ is associated with the interval $[\ell'_l, \ell''_r)$, where $\ell'$ is the left endpoint of the leftmost leaf and $\ell''$  is the right endpoint of the rightmost leaf of the subtree rooted at $w$. 
    Furthermore, internal nodes are also augmented with some profit value $\pi_w$.
    We define query $T(t)$ to return the sum of profit values stored in the nodes in the path going from the root of $T$, down to the leaf $\ell$ such that $t \in [\ell_l, \ell_r)$.

    After initialisation of $T$, we will for each vertex iteratively update $T$.  
    The loop invariant that will be maintained in our algorithm, is that after the $i$th iteration for any time $t \in [0, d_{max} - c]$, $T(t)$ will give the profit collected by an optimal walk ending in vertex $v_i$ at time $t + c$, where $c$ is the cost of a shortest walk travelling from $v_1$ to $v_i$.
    When the aforementioned invariant is maintained, $T(t)$ is always monotonically increasing as the query returns the profit of $W_{i, t}$.
    
    Moreover, the algorithm will maintain the values $\Pi_{max}$ and $c$. Value $\Pi_{max}$ is used to keep track of and return the profit of an optimal walk. Cost $c$ will be increased iteratively, such that at iteration $i$, it stores the cost of the shortest walk that goes from $v_1$ to $v_i$.
    
    Starting with $i = 1$, the only critical time points are $t = 0$ and $t = r_{1, 1}$ because waiting for longer than $r_{1, 1}$ will give us not more than $\pi(v_1)$ profit.
    Hence, for some small constant $\eps > 0$, we can initialise $T$ with two leaves, i.e.~intervals $[0, r_{1, 1})$ and $[r_{1, 1}, d_{max}+\eps)$, storing profits $0$ and $\pi(v_i)$ respectively. 
    Then, trivially, for any time $t \in [0, d_{max}]$, $T(t)$ will return the optimal profit for a walk ending at $v_1$ at time $t$. Thus, we initialise $\Pi_{max} := T(d_{max})$ and $c := 0$.

    Then, for each $1 < i \leq n$, let $c_i$ be the cost of the edge between $v_{i-1}$ and $v_i$, we update $c := c + c_i$.
    For each $1 \leq j \leq m_i$, update $r'_{i,j} = r_{i, j} - c$ and $d'_{i, j} = d_{i, j} - c$.
    Then, locate the leaves $l_a$ and $l_b$ in $T$, such that $r'_{i, j} \in [a_l, a_r)$ and $d'_{i, j} \in [b_l, b_r)$, respectively.
    Thereafter, we remove $[a_l, a_r)$ from $T$ and insert $[a_l, r'_{i, j})$ and $[r'_{i, j, a_r})$ in $T$. Where in both leaves the same profit will be stored as was stored in $l_a$.
    Then, while $T(b_l) + \pi(v_i) \geq T(b_r)$, locate the leaf where $b_r \in [b_r, b_r')$ and remove this leaf from $T$. Update $b_r := b_r'$ and repeat until either $T(b_l) + \pi(v_i) < T(b_r)$ or $b_r = B + \eps$.
    We can then identify the $\mathcal{O}(\log m)$ subtrees that cover the range $[r_{i, j}, b_r)$ and to the profit values of each root of these subtrees we add $\pi(v)$.

    After all time windows of $v_i$ have been processed, we update $\Pi_{max} := \max(\Pi_{max}, T(d_{max} - c))$. Then, after $n$ iterations, we simply can return $\Pi_{max}$. 

    We will now show the correctness of the presented algorithm.
    Let $T$ and $T'$ be the balanced search trees before and after the $i$th iteration, respectively.
    Furthermore, let $c$ be the distance between $v_1$ and $v_{i-1}$ and $c'$ the distance between $v_1$ and $v_i$.
    We assume that for $t \in [0, d_{max} - c]$, $T(t)$ gives the profit of an optimal walk ending in the $v_{i-1}$th vertex at time $t + c$.
    We will show by the use of contradiction that, assuming the loop invariant holds for $T$, the loop invariant also holds for $T'$.

    Consider the walk $W_{i, t' + c}$ that collects $\OPT$ profit and ends in vertex $v_i$ at time $t' + c$. However, assume that $T'(t') \neq \OPT$. We will consider two cases:
    \begin{itemize}
        \item $\exists_{1\leq j\leq m_i} t' + c \in [r_{i, j}, d_{i, j}]$. Since $T$ is monotone in the profit and the profit of $v_i$ is always collected: $\OPT = T(t - c_i) + \pi(v_i) = T'(t)$.
        \item $\nexists_{1\leq j\leq m_i} t \in [r_{i, j}, d_{i, j}]$. Then, there are two possibilities for $W_{i, t' + c}$, either $W_{i, t' + c}$ collected $v_i$ or not. If $v_i$ was collected, then $W_{i, t' + c}$ must have been at $v_i$ at time $d_{i, j'}$, where $d_{i, j'}$ is the last deadline before time $t$. Therefore, again because of the monotonicity of $T$, $\OPT = T(d_{i, j'} - c_i) = T'(t)$. However, if $v_i$ has not been collected, then $\OPT = T(t-c_i) = T'(t)$.
    \end{itemize}

    Hence, the loop invariant is true at initialisation and is maintained throughout the algorithm.
    Furthermore, for every time window, only a constant number of insert, delete, and search operations are performed in the range tree, each of which costs $\mathcal{O}(\log m)$. Thus, the optimal profit can be found in $\mathcal{O}(m \log m)$ time.
\end{proof}

\subsection{Directed Cycle}\label{Sec:OPTWdir}


So far we have seen that
on an undirected path, the \OP with time windows is NP-hard, while on a directed path, it is solvable in polynomial time (regardless of the number of time windows). 
What seemingly makes the problem for directed path easier is that we cannot revisit vertices.
A natural question now is, what happens for directed cycles. In this graph class, a walk
can revisit vertices. But in contrast to undirected path, decisions are more limited because we can only walk in one direction (or wait at a vertex).
We will see that this already makes the problem more complicated to solve than for directed path. Therefore, we start with the Covering \OP (COP) with time windows, which asks if there is a walk such that the profit of every vertex can be collected. 
\subsubsection{Single Time Windows} \label{sec:DCstw}


We start by considering the COP-1TW model on directed cycles, i.e.~every vertex $v_i$ has exactly one time window $[r_i, d_i]$. Recall that a COP-1TW instance is a yes-instance if and only if there exists a walk where all profits are collected.
\begin{lemma}\label{lem:dmax2n}
    Given an instance $I$ of the COP-1TW, there is a reduced instance $I'$ where the latest deadline is upper bounded by $4nC$ (where $C$ is the length of cycle), such that $I$ is a yes-instance if and only if $I'$ is a yes-instance.
\end{lemma}

\begin{proof}
    Given an instance $I$ of the COP-1TW, we encode $I$ as a sequence of events as follows. 
Take all release times and deadlines of instance $I$ and sort them in ascending order. 
Let $\{t_1, t_2, \cdots, t_{2n}\}$ be the resulting sequence, denoted as the  \emph{event sequence} of $I$. 
Obtain the event sequence  $\{t'_1, t'_2, \cdots, t'_{2n}\}$  of the reduced instance $I'$ as follows:
Start with the first event and check each consecutive pair $t_i, t_{i+1}$ (these can be either release times or deadlines) for $i <  2n-1$.
If $t_{i+1} - t_i \leq 2C$, then set $t'_i = t_i$ and $t'_{i+1} = t_{i+1}$.
Otherwise, if  $t_{i+1} - t_i > 2C$, then shrink their gap and update $t'_j := t_j - (t_{i+1} - t_i - 2C)$ for $i <j \leq 2n$. 
Note that after applying this for every of the $2n$ consecutive pairs, the last event (the latest deadline) is at most $4nC$.

Since the order of the events remains the same after the reduction, we get the following key observation:
For each vertex $v_j$ and each event gap $[t_i, t_{i+1}]$, the time window of $v_j$ includes the interval $[t_i, t_{i+1}]$ if and only if the time window of $v_j$ in instance $I'$ includes the interval $[t'_i, t'_{i+1}]$. 

Assume $I$ is a yes-instance and $S$ is a feasible solution of $I$. 
We now show how to construct a feasible solution $S'$ for the instance $I'$. 
For each consecutive pair of events $[t'_i, t'_{i+1}]$ for $1 \leq i \leq n$, if their gap has length less than $2C$, $S'$ takes the route in $S$ between $[t_i, t_{i+1}]$.
Observe that any vertex which has been collected in $S$ between $[t_i, t_{i+1}]$ will be visited and collected between $[t'_i, t'_{i+1}]$ in $S'$.
We consider event gaps of length $2C$.
Let $v_{S, t_i}, v_{S, t_{i+1}}$ be the vertices visited in $S$ at time $t_i$ and $t_{i+1}$, respectively. 
For each event gap $[t'_i, t'_{i+1}]$ of length $2C$, $S'$ takes a non-stop walk starting at 
$v_{S, t_i}$, back to $v_{S, t_i}$ and ending at $v_{S, t_{i+1}}$.
Note that every vertex is visited at least once in $[t'_i, t'_{i+1}]$ and thus every vertex collected by $S$ in $[t_i, t_{i+1}]$ is collected by $S'$ in $[t'_i, t'_{i+1}]$. 
Overall, $S'$ is a feasible solution for collecting all of the vertices in $I'$.

The other direction is analogous.
Let $S'$ be a feasible solution of $I'$. 
For sequential event pairs $t'_i$, $t'_{i+1}$, take the walk between $[t_i, t_{i+1}]$ in $S$ and add extra waits at $t_{i+1}$ if the interval $[t_i, t_{i+1}]$ has length greater than $2C$.
By doing so, each collected vertex collected in $[t'_i, t'_{i+1}]$ is visited and collected in $[t_i, t_{i+1}]$.
\end{proof}

Using Lemma \ref{lem:dmax2n} the following theorem can be proven.

\begin{theorem}\label{thm:cycTSP}
    Given an instance of the COP-1TW on a directed cycle $G$ with arbitrary profits and edge weights, the COP-1TW can be solved in $\mathcal{O}(n^4)$.
\end{theorem}

\begin{proof}
Let $C$ be the total length of the directed cycle.
W.l.o.g.~by Lemma \ref{lem:dmax2n}, assume that the latest deadline $d_{max} \leq 4nC$. We initialise a \emph{schedule} $T$ containing for every $1 \leq i \leq n$ and $1 \leq j \leq 4n$ a value $T_{i, j}$. Time $T_{i, j}$ is the time until which the corresponding walk $W_T$ visits the $i$th vertex for the $j$th time. There is a natural order on the tuples: we define $i', j' \lhd i, j$ as the lexicographic order first on $j'$ and $j$ then on $i'$ and $i$.
Furthermore, for each $1 \leq i, j \leq n$ let $d(v_i, v_j)$ be the distance of the shortest walk travelling from $v_i$ to $v_j$.
Initially, $T$ is a schedule, where $W_T$ never waits, i.e.~$T_{i, j} := d(v_1, v_i) + C \cdot (j-1)$ (see Figure \ref{fig:alg_examp:1}).
We call $T$ a \emph{valid} schedule, if the walk $W_T$ \emph{hits} all intervals, i.e.~$W_T$ visits all the vertices within their time intervals.
Moreover, $W_T$ is called \emph{proper} if the walk is monotone with respect to time and always either waits at vertices or follows the direction of the cycle.
We will modify $T$ iteratively until it is either valid or a conflict occurs, meaning there is no valid schedule.
Apply the following modification iteratively: First, schedule $T$ is tested to see if it is valid. If $T$ is valid, $W_T$ can be returned, otherwise the algorithm continues with the following steps. 
If there exists a vertex $v_i$ with time window $[r_i, d_i]$ that walk $W_T$ does not hit, then we first check if $T_{i, 1} > d_i$. If $T_{i, 1} > d_i$, then we return that no valid solution exists. 
Otherwise, set 
\begin{align}
    &j := \max\{j | T_{i, j} < r_i\} \text{~~~~~and} \label{eq:chooseJfullproof} \\
    &T_{i, j} := r_i. \label{eq:upTfullproof}
\end{align}
Further, for all $T_{i', j'}$ with $i, j \lhd i', j'$ set \begin{equation}
    T_{i', j'} := \max(T_{i', j'}, r_i + d(v_i, v_{i'}) + C \cdot (j' - j)).
\end{equation} 

Then, we start again at the beginning of this loop and test if there exists a vertex with a time window that is not currently hit.

\begin{figure}[t]
     \centering
     \begin{subfigure}[t]{0.32\textwidth}
         \centering
         \includegraphics[page=4]{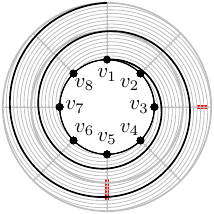}
         \caption{The walk $W_T$ directly after initialisation of schedule $T$ which does not wait.}
         \label{fig:alg_examp:1}
     \end{subfigure}
     \hfill
     \begin{subfigure}[t]{0.32\textwidth}
         \centering
         \includegraphics[page=5]{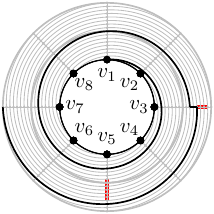}
         \caption{Vertex $v_3$ is not collected, hence time $T_{3, 2}$ was increased to $r_3 = 12$ and the rest of the schedule is updated sequentially.}
         \label{fig:alg_examp:2}
     \end{subfigure}
     \hfill
     \begin{subfigure}[t]{0.32\textwidth}
         \centering
         \includegraphics[page=6]{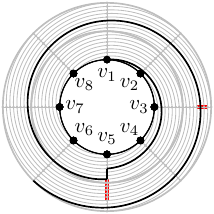}
         \caption{Now vertex $v_5$ is no longer collected, hence time $T_{5, 1}$ was increased to $r_5 = 7$.}
         \label{fig:fig:alg_examp:3}
     \end{subfigure}
        \caption{An example of two iterations of the algorithm described in Theorem \ref{thm:cycTSP}. The temporal graphs illustrate possible walks on a clockwise directed cycle with $n = 8$ and unit edge weights, edges in the temporal graph cannot be traversed in counterclockwise order. 
        }
        \label{fig:algoExample}
        \vspace{-5mm}
\end{figure}

Two iterations of the algorithm are illustrated in Figure \ref{fig:algoExample}.
To show that the algorithm is correct we show that the algorithm maintains the following loop invariant: Walk $W_T$ is a proper walk on the cycle and, for each $1 \leq i, j \leq n$, $T_{i, j} \leq S_{i,j}$ for any valid schedule $S$.
In the first round of the algorithm, the invariant trivially holds; walk $W_T$ is proper by definition and does not wait, thus if $S$ exists, then for every $1 \leq i, j \leq n$, $T_{i, j} \leq S_{i,j}$.

Assume that the loop invariant holds at the beginning of an iteration and let there be some $v_i$ for which $T_{i, 1} > d_i$. Since, the loop invariant is true we know that $d_i < T{i, 1} \leq S_{i,j}$ for any valid schedule. This inequality is only possible if no valid schedule exists, as it directly contradicts the definition of a valid schedule. 
Thus, the algorithm is correct if it returns that there is no valid schedule.

Next, we show that after updating, $W_T$ is still a proper walk on the cycle. The schedule remains unchanged for all $i', j' \lhd i, j$, thus up to time $T_{i, j}$ the path remains proper. 
Then, updating $T_{i, j}$ is a safe operation as $r_i > T_{i, j}$, which means that $W_T$ will simply wait at vertex $v_i$ in the $j$th round. 
Furthermore, for all $i, j \lhd i', j'$ the last visit times are updated based on the minimum travel time from $v_i$ to $v_{i'}$ with $j' - j$ rounds in between. 
Thus, $W_T$ remains proper since either $T_{i', j'}$ is the minimum travel time from leaving $v_i$ at $r_i$ or a larger value. 
As soon as, $T_{i', j'} \geq r_i + d(v_i, v_{i'}) + C \cdot (j' - j)$, for all $i', j' \lhd i'', j''$, value $T_{i'', j''}$ remains unchanged because this part was a proper path and therefore must be greater than or equal to the minimum travel time from $v_i$.
Thus, after one update cycle, walk $W_T$ is still proper.

Second, let $T'$ be the schedule after $T$ has been updated, we need to show that after an iteration the loop invariant holds for schedule $T'$ assuming the invariant holds for $T$.
Trivially, the loop invariant holds for the schedule up to $T'_{i,j}$, since this part has not been modified.
Since, the loop invariant is true at the start of an iteration and $v_i$ is not collected, we know $d_i < T_{i, j+1} \leq S_{i, j+1}$.
So for any valid solution, vertex $v_i$ must be collected strictly before the $j + 1$th round, thus $S_{i, j} \geq r_i = T'_{i, j}$.
Then, for all $i, j \lhd i', j'$, where $r_i + d(v_i, v_{i'}) + C \cdot (j' - j) > T_{i', j'}$, since $S$ corresponds to a proper walk $W_S$, $S_{i', j'} \geq r_i + d(v_i, v_{i'}) + C \cdot (j' - j) = T'_{i', j'}$.
Finally, for all $i, j \lhd i'', j''$, where $r_i + d(v_i, v_{i''}) + C \cdot (j'' - j) \leq T_{i'', j''}$, it holds that $T_{i'', j''} = T'_{i'', j''}$ and hence $S_{i'', j''} \geq T'_{i'', j''}$.

Thus, by our loop invariant, termination occurs only if no valid schedule exists or a valid schedule is found. 
Due to the choice of $j$ in Equation \ref{eq:chooseJfullproof}, for every $1 \leq j \leq 4n$, the value of $T_{i, j}$ can be updated at most once in Equation~\ref{eq:upTfullproof}. Thus, after $\mathcal{O}(n^2)$ iterations, the algorithm must either report a valid schedule or that no such schedule exists. 
Furthermore, each iteration takes at most $\mathcal{O}(n^2)$ steps. Thus, a valid schedule will be computed in $\mathcal{O}(n^4)$ time if it exists.
\end{proof}

Thus, COP-1TW can be solved in polynomial time.
In Section \ref{sec:DCdtw}, we will see that the version with multiple time windows, COP-MTW and thus OP-MTW is NP-hard. It remains open for future work, whether OP-1TW NP-hard. However, in this section a 2-approximation algorithm is given using the following theorem: 

\begin{theorem}\label{thm:cycSingle}
Given an instance of the OP-1TW on a directed cycle $G$ with arbitrary profits and edge weights, and where every time window is shorter than the total length of the cycle, the OP-1TW can be solved in $\mathcal{O}(n^2 \log n)$ time.
\end{theorem}

\begin{proof}
    Let $C$ be the cost of traversing cycle $G$ exactly once. W.l.o.g.~assume that $d_{max} \leq 4nC$ (Lemma \ref{lem:dmax2n}).
    To solve the \OP on a directed cycle with `short' time windows the cycle can be unwrapped and concatenated to itself $4n$ times, resulting in a directed path $G'$. On this directed path the result of Lemma \ref{prop:dirOP} can be applied.
    
    More formally, for $1 \leq j \leq 4n$ and for each $1 \leq i \leq n$ create a vertex $v_{i, j}$ and add this to $V(G')$ with profit $\pi(v_i)$ and time window $[r_i, d_i]$, where $v_i \in V(G)$. Furthermore, for $1 \leq j \leq 4n$ and $1 \leq i < n$, add an edge $(v_{i, j}, v_{i+1, j})$ to $E(G')$ with the same cost as $(v_i, v_{i+1}) \in E(G)$. Finally, for $1 \leq j < 4n$, add an edge $(v_{n, j}, v_{1, j+1})$ to $E(G')$ with the cost of $(v_n, v_1) \in E(G)$.
    Let $\OPT$ be the profit of an optimal walk $W$ in $G$ and $\OPT'$ be the profit of an optimal walk $W'$ in $G'$. Then it, always holds that $\OPT = \OPT'$.
    
    Since all time windows are shorter than $C$, walk $W'$ could not have collected two distinct vertices $v_{i, j}$ and $v_{i, j'}$ for some $1 \leq i \leq n$ and $1 \leq j, j' \leq n$. Therefore, any walk in $G$ can be modified to collect the same profit in $G'$ and vice versa. 
    Hence, Lemma \ref{prop:dirOP} can be used to solve $G'$, which can directly be used as a solution for $G$. Since, $G'$ has $4n^2$ intervals, the solution to the OP-1TW on cycles with no time windows longer than or equal to $C$ can be calculated in $\mathcal{O}(n^2 \log n)$.
\end{proof}

The problem why a dynamic programming algorithm like the one in Theorem~\ref{thm:cycSingle} does not extend to `long' intervals is that when the algorithm encounters such an interval, it does not know whether it already collected the profit in a previous iteration. If all but $k$ intervals are `short', the algorithm could of course, at the cost of a factor of $2^k$ in the running time, keep track on the profit depending on which subset of the $k$ `long' intervals it already collected. Alternatively, we can consider `short' and `long' intervals separately, resulting in a 2-approximation.

\begin{corollary}\label{cor:cycleFPT}
    Given an instance of the OP-1TW on a directed cycle with arbitrary profits and edge weights, the OP-1TW is fixed-parameter tractable in the number of time windows that are at least as long as the length of the cycle.
\end{corollary}

\begin{corollary}\label{cor:cycleApprox}
    Given an instance of the OP-1TW on a directed cycle $G$ with arbitrary profits and edge weights, there is a 2-approximation algorithm that runs in $\mathcal{O}(n^2\log n)$ time.

\end{corollary}

\begin{proof}
    Let $C$ be the total cycle length of $G$.
    Furthermore, let $G_s$ and $G_l$ be exact copies of $G$.
    We will modify $G_s$ such that it only has `short' time windows.
    More specifically, for $1 \leq i \leq n$ if $d_i - r_i \geq C$, in $G_s$ set $r_i = d_i = 0$ and $\pi(v_i) = 0$. 
    Furthermore, we will modify $G_l$ such that it only has `long' time windows.
    Analogously, for $1 \leq i \leq n$ if $d_i - r_i < C$, in $G_l$ set $r_i = d_i = 0$ and $\pi(v_i) = 0$. 

    Using Lemma \ref{thm:cycSingle}, we can solve the instance corresponding to $G_s$ obtaining profit $\OPT_s$.
    Furthermore, since all vertices in $G_l$ with profit have time windows longer than $C$, a walk which continues until $d_{max}$ collects all available profit $\OPT_l$.

    Now let $\OPT$ be the profit of an optimal walk in $G$, then clearly $\OPT \leq \OPT_s + \OPT_l \leq 2 \max(\OPT_s, \OPT_l)$. Hence, $\max(\OPT_s, \OPT_l)$ is a 2-approxi\-mation that can be computed in $\mathcal{O}(n^2 \log n)$ time.
\end{proof}

In the following we develop a polynomial-time approximation scheme for the OP-1TW. 
For this, we first consider (sub-)problems that can be optimally solved by a walk taking at most $k$ rounds. A \emph{round} is a walk starting and ending at $v_0$ visiting every other vertex once. 
The following lemma makes use generalization of the dynamic program as used in Proposition~\ref{prop:dirOP} for directed paths.

\begin{lemma}\label{lem:k-rounds}
    Given an instance of the OP-1TW on a directed cycle with arbitrary profits and edge weights where a walk can take at most $k$ rounds, the OP-1TW can be solved in $n^{\mathcal{O}(k)}$ time.
\end{lemma}

\begin{proof}
    In the following we assume that we know the starting time $s_i$ of every round $i$ in the optimal walk (which can by guessed by trying all possible starting times with an additional factor of ${\mathcal{O}(n^k)}$ time). An analogous dynamic program as used in Proposition \ref{prop:dirOP} can be used, except that this dynamic program has $k+1$ parameters $j, t_1, \cdots, t_k$. An entry in the dynamic program will contain the maximum profit that can be collected from the first $j$ vertices $v_1, \cdots, v_j$, 
    assuming the path starts round $i$ at time $s_i$ and leaves $v_j$ at time $t_i$.

    Then, for $j+1$ and a set of $t_i$, we can compute the corresponding entry based on the solutions for $j$. Since, the edge weights can be arbitrary, the set of $t_i$ must contain all the points at which an optimal path can leave $v_{j+1}$.
    We can assume an optimal path only waits at a vertex until its profit will be collected (with an argument similar to that used in Theorem \ref{thm:cycTSP}). 
    Then, time $t_i$ either must be $s_i + d(v_1, v_j)$ or $r_k + d(v_k, v_j)$ for all $1 \leq k \leq j$, where $d(v_k, v_j)$ gives the minimum time necessary to travel from $v_k$ to $v_j$. 

    To conclude, the dynamic program has $n^{\mathcal{O}(k)}$ possible entries, and each entry can be computed in $\mathcal{O}(nk)$ time based on previous entries. Hence, the algorithm runs in $n^{\mathcal{O}(k)}$ time.
\end{proof}


Now suppose that we would compute for every pair of times $t < t'$ the maximum-profit walk that takes at most $k$ rounds. Unfortunately, concatenating such walks is not optimal, since they might collect the profit of the same vertices. To circumvent this problem we restrict to walks that occasionally perform a \emph{sprint}, that is, a round in which the walk does not wait at vertices. A sprint takes exactly $C$ time. 
We define such a walk as a \emph{$k$-sprint}: a sequence of at most $k$ rounds of which the last round is a sprint. Furthermore, we define a \emph{$k$-workout} as a walk consisting of a sequence of $k$-sprints, meaning that every $k$ rounds there is at least one sprint.

\begin{lemma}\label{lem:dc-ptas2}
    Given an instance of the OP-1TW on a directed cycle with arbitrary profits and edge weights, the optimal $k$-workout can be computed in $n^{\mathcal{O}(k)}$ time.
\end{lemma}

\begin{proof}
    For every pair of time $t < t'$ we will compute the following: the optimal $k$-sprint starting at $t$ and ending at $t'$, but ignoring profits of vertices where (a) time windows don't intersect $[t, t']$ or (b) the profit would have been collected by a sprint ending at time $t$.

    This optimal $k$-sprint can be computed with the algorithm of Lemma~\ref{lem:k-rounds} by limiting the overall time to $[t, t']$ and the last round to a sprint (which can simply be achieved by first computing the profits collected by the sprint and removing those profits from the instance).

    We now consider a directed graph with edges $(t, t')$ and as edge weights the profit of the optimal $k$-sprint for $t, t'$ as computed above. Then the weight of the maximum-weight path is the profit of the optimal $k$-workout. For this, two conditions must be checked: Firstly, no profit is collected twice. Secondly, by not considering all profits when computing the profits of the $k$-sprints, no profits were missed.

    Assume some profit of vertex $v$ is collected more than once. Consider any but the first $k$-sprint during which this profit is collected. Since there is a previous $k$-sprint that collects the profit of $v$, the time window of $v$ spans back far enough so that it would have been collected by a sprint just before the $k$ sprint. However, according to (b) above, the profit would have been excluded from consideration. Hence, profits are collected at most once.
    
    Now consider an optimal $k$-workout. Split it into $k$-sprints. This now corresponds to a path in our graph, i.e., each of the $k$-sprints corresponds to an edge. We argue that the weight of an edge computes the profit of a $k$-sprint correctly. The profits that we ignored during the computation of the edge weight, could not have been by the $k$-sprint. In particular those of case (b) couldn't have been collected by that $k$-sprints, because they are already collected by a previous $k$-sprint.
\end{proof}


Using Lemma \ref{lem:dc-ptas2} we can then construct a PTAS for the OP-1TW on a directed cycle:


    


\begin{theorem}\label{thm:OPTWdircyc-ptas}
    Given an instance of the OP-1TW on a directed cycle $G$ with arbitrary profits and edge weights, there is a $(1 + \eps)$-approximation algorithm that runs in $n^{\mathcal{O}(1/\eps)}$ time.
\end{theorem}

\begin{proof}
    Consider an optimal walk $W$ collecting $\OPT$ profit and split $W$ into sections of $k$ rounds. In each of these sections, replace the least profitable round by a sprint plus waiting time, giving us walk $W'$. Then, the maximum number of rounds between two sprints in $W'$ is $2k-2$, which means that $W'$ is a $2k$-workout. Furthermore, in the worst case, the profit is reduced at most by a factor of $(k-1)/k$.

    Hence, by selecting $k-1 = \lceil 1/\eps \rceil$, we obtain a $(1 + \eps)$-approximation algorithm that runs in $n^{\mathcal{O}(1/\eps)}$ time.
\end{proof}


\subsubsection{Multiple Time Windows} \label{sec:DCdtw}
    We show that the COP-MTW on directed cycles is NP-complete, implying NP-completeness for the OP-MTW.

\begin{theorem}\label{thm:cycMTW}
    The COP-MTW on directed cycles is NP-complete.
\end{theorem}

\begin{proof}[Sketch]
     Clearly this problem lies in NP as we can non-deterministically guess a permutation of vertices and check if we can collect all vertices according to this order.
     
     In the following, we show the NP-hardness by  a reduction from the well-known NP-complete problem 3SAT.  
    Consider a 3SAT instance with a set $X$ of $n$ variables and a set $\Gamma$ of $m$ clauses. 
    In 3SAT, each clause $\gamma \in \Gamma$ has exactly three literals, and the goal is to find an assignment of $X$ such that the formula is satisfied. 
    We reduce this problem to the COP-MTW on a directed cycle with $C = n + m$ vertices.
    Since in this proof for the time windows only discrete points instead of multiple intervals are needed, time windows will be described by sets of time points. 
    For every $1 \leq i \leq n$, each variable $x_i \in X$ corresponds to $v_i$ in the cycle with time points $\{2iC, 2iC + C - 1\}$.  
    Furthermore, for every $1 \leq j \leq m$, clause $\gamma_j \in \Gamma$ corresponds to vertex $v_{n + j}$ with an initial empty time window $\{\}$ (see Figure \ref{fig:DCdtw}).
    Then, for every clause $\gamma_j$ in which $x_i$ appears as a positive literal, the time point $2iC + (n + j - i)$ is added to vertex $v_{n+j}$. Otherwise, if $x_i$ appears as a negated literal in $\gamma_j$, the time point $2iC + C - (m - j + i) - 1$ is added to $v_{n + j}$ (see Figure~\ref{fig:DCdtw} for an example).
    Note that in this construction, no profit can be collected at time $t \in [iC, 2iC)$ for any $0 \leq i \leq n$.

     \begin{figure}
        \vspace{-3mm}
         \centering
         \includegraphics{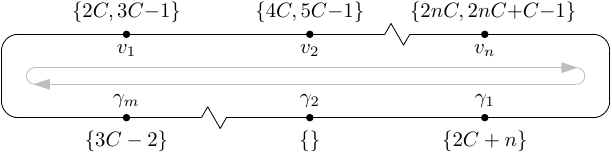}
         \caption{Example reduction from 3SAT, where $v_1$ appears in $\gamma_1$ as a positive and in $\gamma_m$ as a negated literal.}
         \label{fig:DCdtw}
        \vspace{-5mm}
     \end{figure}

    We will now show that in the time range $[2iC, 3iC)$ it is not possible to collect the vertices corresponding to the variable $x_i$ and clauses $\gamma_j, \gamma_{j'}$ where $x_i$ appears as a positive and negative literal respectively. First of all, to collect $x_i$ we must be at vertex $v_i$ at time $2iC$ or $2iC + C - 1$. 
    In the case that the $v_i$ is collected at time $2iC$, then it is possible to collect the vertex corresponding to $\gamma_j$ ($v_{n + j}$), since the distance between these vertices is exactly $n + j - i$. However, in this case it is not possible to collect $v_{n + j'}$: On one side, $v_i$ must be collected before $v_{n+j'}$ as $2iC + C - (m - j' + i) - 1 \geq 2iC$. However, the $v_{n + j'}$th vertex cannot be collected after $v_{i}$ either, since the distance between these vertices is exactly $n + j' - i$. This means that the earliest time $v_{n + j'}$ can be reached is at $2iC + n + j' - i = 2iC + C - n - m + n + j' - i = 2iC + C - (m - j' + i) > 2iC + C - (m - j' + i) - 1$. Hence, if $v_i$ is collected at time $2iC$, only the vertices corresponding to clauses where $v_i$ appears as a positive literal can be collected.

    In the case that $v_i$ is collected at time $2iC + C - 1$, then the vertex corresponding to $\gamma_{j'}$ can be collected at an earlier time, since the distance is exactly $m - j' + i$. Furthermore, analogous to the previous case, we know that $2iC + (n + j - i) \leq 2iC + C - 1$, so $v_{n+j}$ must be collected before $v_i$. However, this distance is $m - j + i$ and $2iC + (n + j - i) + (m - j + i) = 2iC + C > 2iC + C - 1$. Therefore, it is never possible to collect the vertices corresponding to $v_i$, $\gamma_j$, and $\gamma_{j'}$ between the times of $2iC$ and $3iC$.
     
    Then, if there exists a feasible solution in the 3SAT instance, the assignment could be used to choose when to visit the vertices corresponding to the variables. Since every clause is satisfied, this means that in our instance every clause vertex can be reached from at least one vertex variable, giving a valid walk that collects all vertices. Analogously, if a valid walk exists in the COP-MTW instance, we can use the time of collecting the vertices corresponding to the variables for an assignment of the 3SAT instance.
\end{proof}

\section{The \OP on Dynamic Graphs} \label{sec:dyngraphs}

In the \OP with time windows, profits are only collectable in certain time windows assigned to the corresponding vertex. 
A similar extension of the \OP is to add time windows to the edges and thus consider the \OP in a setting with dynamic graphs. For an instance of the \OP on dynamic graphs with input graph $G=(V,E)$, the time windows for edges are sets of intervals $\mathcal{I}_e$ for all $e \in E$. The graph at time step $i$ is then $G_i=(V,E_i)$ with $E_i=\{e \in E \mid \exists I \in \mathcal{I}_e \text{ s.t. } i \in I \}$. We say an edge $e$ is \emph{active} at time step $i$, if $e \in E_i$.

In contrast to the \OP with time windows, directed paths or cycles can be solved trivially by simply advancing the walk whenever possible. Surprisingly, even dynamic undirected paths can be solved with a simple algorithm, as shown in Proposition \ref{lem:dyn_path}. However, the problem immediately becomes NP-hard when dynamic trees are considered, as shown in Proposition~\ref{prop:dyn_trees}.

\begin{proposition}\label{lem:dyn_path}
Given an instance of the \OP with a dynamic undirected path $(V, E)$, this instance can be solved in $\mathcal{O}(n^2)$ time.
\end{proposition}

\begin{proof}
     Let $V(G) = \{v_1, \cdots, v_n\}$ with $E(G) =\{\{v_i, v_{i+1} \} \mid 1 \leq i < n\}$, where $v_\sigma = s$.
    Let $\OPT$ be the profit collected by an optimal walk $W$. Let $v_i$ and $v_j$ be vertices such that $1 \leq i \leq \sigma \leq j \leq n$ and $i \leq v_w \leq j$ for all $v_w \in W$.
    A shortest optimal walk would start at $s=v_\sigma$, walk as directly as possible to $v_i$ ($v_j$ respectively) and then to $v_j$ ($v_i$ respectively). 
    Thus, by guessing $v_i$ and $v_j$ we obtain an optimal walk in $\mathcal{O}(n^2)$ time.
\end{proof}


\begin{proposition}\label{prop:dyn_trees}
The \OP is NP-hard even on dynamic trees.
\end{proposition}

\begin{proof}
    We will prove that the \OP is NP-hard on dynamic trees by a reduction from the 3-partition problem. Given a multiset $I$ of $n = 3m$ integers, with total sum $mT$, i.e.~$\sum_{i\in I} s = mT$, can $I$ be partitioned into a set of triplets $I_1, \cdots, I_m$ such that the sum of each triplet is exactly $T$?

    Given such a 3-partition instance, we will create a new \OP instance on a dynamic graph $(V, E)$. 
    The constructed graph $(V,E)$ will be a spider graph with start node $s$ at the centre. 
    Then, for each integer $i \in I$, we connect a leg of length $i$ to $s$, whose edges are active at times $[0, 2mT + 2m]$. 
    Next, we add $m$ control vertices $c_1, \cdots, c_m$ with edges to $s$, such that the edge incident to $c_i$ is active at $[2iT + 2(i-1), 2iT + 2(i-1) + 2]$. 
    An example of such a graph can be seen in Figure \ref{fig:examplereduc}.
    
    \begin{figure}
        \centering
        \includegraphics{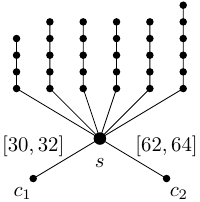}
        \caption{Example reduction for the instance $\{4, 5, 5, 5, 5, 6\}$, making $T = 15$. All upper edges are active at time $t \in [0, 2mT + 2m]$.}
        \label{fig:examplereduc}
    \end{figure}
    
    What remains is to show that it is possible to collect exactly $m + mT + 1$ profit in our \OP instance if and only if the 3-partition instance is a yes-instance. 
    First, assume walk $W$ exists which collects exactly $m + mT + 1$ profit. Note that this is exactly the number of nodes in $V$, which means that we visit every control vertex $c_i$ at exactly $2iT + 2(i-1) + 1$. Furthermore, since the intervals of all other vertices expire at $2mT + 2m$, we know that each edge is used at most twice. Hence, walk $W$ cannot first partially collect a leg and later return to this leg and collect the remainder. 
    Thus, walk $W$ must repeatedly collect exactly 3 legs with profit exactly $T$ and then collect a control vertex (there cannot be more than 3 legs collected between two control vertices, since every integer in $I$ is strictly between $T/4$ and $T/3$). 
    Partitioning $I$ based on the legs visited between two control vertices gives a valid solution to the 3-partition instance.
    
    Now suppose we can partition $I$ into triplets $I_1, \cdots, I_m$ such that the sum of each triplet is exactly $T$. Then there exists a walk in $(V,E)$, where we first visit the legs corresponding to the integers of $I_1$, this takes exactly $2T$ time steps and collects $T$ profit. Thereafter, we visit the first control vertex $c_1$, whose incident edge is active in interval $[2T, 2T + 2]$, and go back to $s$, resulting in 1 additional profit and 2 additional time steps. We repeat these steps for every interval $I_i$, each time collecting the corresponding legs and then visiting the $i$th control vertex. Therefore, this walk collects exactly $m + mT + 1$ profit. 
\end{proof}




\section{The Orienteering Problem without time windows} \label{sec:OP}

In general, the \OP is NP-hard even on trees and thus also on graphs with bounded tree-width. 
If vertex profits are unit, the \OP is still NP-hard in general~\cite{blum2007approximation}, but we show that in this case it is possible to give a dynamic program on a tree decomposition. We will extend this result to a $(1+\eps)$-approximation for the \OP with arbitrary vertex profits on graphs with bounded \tw.  

Note that, it is quite easy to show that the \OP with unit vertex profits is polynomial for a tree, by giving a dynamic program: For every vertex we only need to distinct if it has been passed an even or an uneven number of times. 
Contrary, for graphs with bounded tree-width, we have to add an additional partition of special structures to the dynamic program. 

\begin{lemma}\label{thm:dp-tree}
    Given an instance of the \OP with unit vertex profit and input tree $G$, the \OP is solvable in $\mathcal{O}(n^3)$ time.
\end{lemma}

\begin{proof}
    First, we see start node $s$ as the root of the tree. Furthermore, we assume that the maximum degree in the tree, this can be achieved by transforming $G$ by adding zero-profit vertices connected with zero-cost edges.
    
    Then, we will fill a dynamic program $a[v, \Pi, d]$, where $v \in V$, $0 \leq \Pi \leq n$, and $d \in \{e, u\}$. An entry in this dynamic program will contain the cost of a minimum path which starts at vertex $v$, collects exactly $\Pi$ profit and in which $v$ has even ($e$) or uneven $(u)$ degree. Then, $a$ can be filled by processing the vertices in a topological order and the entries can be filled based on the type of $v$:

    \textbf{Leaf} let $v$ be a leaf in $G$, then $a[v, 0, e] = a[v, 0, u] = +\infty$, $a[v, 1, e], a[v, 1, u] = 0$, and all others $+\infty$
    
    \textbf{Inner node} let $v$ be an inner node in $G$, with children $w_1, w_2$, then
        \begin{itemize}
            \item[] $a[v, i, e] =$ \\
            \item[]
            \hspace*{0.2cm} $\min \left. \begin{cases}
            a[w_1, i{-}\pi(v), e] + 2c(v, w_1),\\
            a[w_2, i{-}\pi(v), e] + 2c(v, w_2),\\
            \displaystyle\min_{1 \leq j \leq i{-}\pi(v)} a[w_1, j, e] + a[w_2, i{-}\pi(v){-}j, e] + 2c(v, w_1) + 2c(v, w_2)
        
            \end{cases}\right\}$\\
            \item[]  
                $a[v, i, u] =$ \\
            \item[]
            \hspace*{0.2cm} $\min \left. \begin{cases}
                a[w_1, i{-}\pi(v), u] + c(v, w_1),\\
                a[w_2, i{-}\pi(v), u] + c(v, w_2),\\
                \displaystyle\min_{1 \leq j \leq i - \pi(v)} a[w_1, j, u] + a[w_2, i{-}\pi(v){-}j, e] + c(v, w_1) + 2c(v, w_2),\\
                \displaystyle\min_{1 \leq j \leq i - \pi(v)} a[w_1, j, e] + a[w_2, i{-}\pi(v){-}j, u] + 2c(v, w_1) + c(v, w_2)
                \end{cases}\right\}$
        \end{itemize}

    the best profit is thus obtained by finding the maximum $1 \leq \Pi \leq n$ such that $c[s, \Pi, e] \leq B$ or $c[s, \Pi, u] \leq B$\
\end{proof}

\begin{lemma}\label{thm:main}
    Given an instance of the \OP with unit vertex profits and input graph $G$ of bounded \tw $k$, the \OP is solvable in $\mathcal{O}(n^2) \cdot k^{\mathcal{O}(k)}$ time.
\end{lemma}

\begin{proof}
    W.l.o.g assume that $\T = (T, \{X_t\}_{t \in V(T)})$ is a nice tree decomposition for $G$ of width $k$ as given in~\cite{cygan2015parameterized}. Note that their definition of a ''nice`` tree-decomposition is slightly different from the commonly used one, which includes \emph{introduce}, \emph{forget}, or \emph{join} node. In~\cite{cygan2015parameterized}, the authors give \emph{introduce edge} and \emph{introduce vertex} nodes, where the introduce vertex nodes correspond to introduce nodes in the standard sense and an introduce edge node $t$ for edge $uv\in E(G)$ is labelled by $uv$, has exactly one child $t'$ such that $X_t=X_{t'}$ and it holds that $u,v \in X_t$. Furthermore, for every edge $uv \in E(G)$, there is exactly one introduce edge node labelled by $uv$ in $\T$. 
    
    Further assume that $s \in X_t~ \forall t \in V(T)$ (increases tree-width by at most $1$) and $X_v = \{s\}$ if $v$ is the root or a leaf. Let $G_t$ for $t \in V(T)$ be the induced subgraph of $G$ with vertex set $\bigcup_{t' \succeq t} X_{t'}$, that is, the subgraph of $G$ containing all vertices in the bags of the subtree of $T$ rooted at $t \in V(T)$.
    
    For a walk $H$ starting at $s$, for some $t \in T$, the part of $H$ contained in $G_t$ is a non-empty set of disconnected walks. Let $F$ be such a set of disconnected walks that collects the most distinct vertices and let $X \subseteq X_t$ such that for any profit $1 \leq \iPi \leq n$ 
    it holds that $s \in V(F)$, $V(F) \cap X_t = X$, $|V(F)| \leq \iPi$, and $\P=\{ X_t \cap C \mid C \text{ is a connected component of } V(F)\}$ is a partition of $X$, where $V(F)$ is the set of vertices in $F$. 
    
    We then define $c[t, X, \P, \Pi]$ as the minimum cost of a possible $F$ in $G_t$.
    
    The best profit is thus obtained by finding the maximum $0 \leq \Pi \leq n$ such that $c[r, \{s\}, \{\{s\}\}, \Pi] \leq B$, where $r\in\T$ such that $G_r = G$.
    If $F$ does not exist, set $c[t, X, \P, \Pi] = +\infty$. Otherwise, fill $c[t, X, \P, \Pi]$ depending on the node type of $t$:
    
    \textbf{Leaf node.} Let $t$ be a leaf node. Then $X_t = \{s\}$, so for each profit $1 \leq \iPi \leq n$, set $c[t, \emptyset, \emptyset, 0] = c[t, \emptyset, \emptyset, \iPi] = c[t, \{s\}, \{\{s\}\}, \iPi] = 0$, and $c[t, \{s\}, \{\{s\}\}, 0] = +\infty$.
    
    \textbf{Introduce vertex node.} Let $t$ be an introduce vertex node with  child node $t'$ s.t.~$X_t = X_{t'} \cup \{v\}$, $v \notin X_{t'}$. As edges are introduced in introduce edge nodes, assume $v$ is isolated in $G_t$. Then for all $X \subseteq X_t$, partition $\P = \{P_1, P_2, \cdots, P_q\}$ of $X$, and $0 \leq \iPi \leq n$: If $v$ is in $X$, then $\{v\} \in \P$, that is $v$ is its own connected component. 
    Thus:
    \begin{align*}
        c[t, X, \P, \iPi] =
        \begin{cases}
        c[t', X \setminus \{v\}, \P \setminus \{\{v\}\}, \iPi - 1] & \text{if $v \in X$,} \\
        c[t', X, \P, \iPi] & \text{otherwise.}
        \end{cases}
    \end{align*}
    
    \textbf{Introduce edge node.}
    Let $t$ be an introduce edge node labelled with $uv$, and let $t'$ be child node of $t$. For each $X \subseteq X_t$, partition $\P = \{P_1, P_2, \cdots, P_q\}$ of $X$, and $0 \leq \iPi \leq n$ we consider four cases. If $u \notin X$ or $v \notin X$, then $\{u,v\} \not\subseteq V(F)$, hence, $c[t, X, \P, \iPi] = c[t', X, \P, \iPi]$. The same holds 
    when $u,v$ are in $X$, but not in the same set of $\P$. 
    Assuming that $u, v \in X$ and $u,v \in P_i \in \P$, edge $uv$ has either been included in the solution or not. If not, we consider the same partition $\P$ at $t'$, otherwise $P_i$ has been obtained by merging two smaller sets of the partition, one containing $u$ and the other $v$. Hence,
    \begin{align*}
        c[t, X, \P, \iPi] = \min \left. \begin{cases}
        \min_{\P' \in \mathbf{P'} } c[t', X, \P', \iPi] + c(uv),\\ \min_{\P'' \in \mathbf{P''}} c[t', X, \P'', \iPi] + 2c(uv), \\c[t', X, \P, \iPi] 
        \end{cases}\right\},
    \end{align*}
    
    where $\mathbf{P'}$ is defined as the set of all partitions of $X$, such that in partition $\P' \in \mathbf{P'}$ vertices $u$ and $v$ are in separate sets of $P'$ and after merging these sets partition $\P$ is obtained. Furthermore, these sets result in a walk where edge $uv$ is traversed exactly once. 
    This 
    is checked by constructing an auxiliary structure $G_{\P'}$, which is a collection of walks where each connected component corresponds to exactly one set of $\P'$.
    A partition $\P'$ is in $\mathbf{P'}$ if and only if, for $u \in P_i \in \P'$, $P_i = \{u\}$ or if $u \neq s$ then $u$ must be of odd degree, otherwise if $u = v$ then $u$ must be of even degree. Furthermore, the constraints for $v \in P_j \in \P'$ are defined analogously. 
    
    Similarly, set $\mathbf{P''}$ is defined as the set of all partitions of $X$, such that in partition $\P'' \in \mathbf{P''}$ vertices $u$ and $v$ are in separate sets of $P''$ and after merging these sets partition $\P$ is obtained. However, these sets result in a walk where edge $uv$ is traversed exactly twice.
    This property does not require additional checking, since any two walks can be merged into a valid walk with an edge that is traversed twice.
    Furthermore, no additional sets of partitions are needed, since in an optimal walk each edge is used at most twice.
    
    \textbf{Forget node.}
    Suppose that $t$ is a forget node with child node $t'$ such that $X_t = X_{t'} \setminus \{ w \}$ for some $w \in X_{t'}$. Consider every $X \subseteq X_t$, partition $\P = \{P_1, P_2, \cdots, \}$ of $X$, and $0 \leq \iPi \leq n$. The solution for $X$ might either use the vertex $w$, in which case it is added to one of the sets in $\P$, or not, in which case the same partition in $t'$ remains optimal. Hence,
    \begin{align*}
        c[t, X, \P, \iPi] = \min \left\{ \min_{\P' \in \mathbf{P'}} c[t', X \cup \{w\}, \P', \iPi + 1], c[t', X, \P, \iPi] \right\},
    \end{align*}
    where $\mathbf{P'}$ is defined as the set of all partitions of $X \cup \{w\}$ where $\P' \in \mathbf{P'}$ is obtained by adding $w$ to one of the sets of $\P$.
    
    \textbf{Join node.}
    Suppose $t$ is a join node with children $t_1$ and $t_2$. 
    For an optimal solution of $G_t$, solutions of $G_{t_1}$ and $G_{t_2}$ need to be merged without creating cycles or destroying walks. The same auxiliary structure $G_{\P'}$ as in \emph{introduce edge nodes} can be used to encode the merging of $G_{t_1}$ and $G_{t_2}$. A partition $\P$ of $X$ is a \emph{proper walk merge} of partitions $\P_1$ and $\P_2$ if the merge of collections of walks $G_{\P_1}$ and $G_{\P_2}$ (treated as a multigraph) is a collection of walks whose family of connected components corresponds to $\P$. 
    Moreover, in the walk containing $s$, vertex $s$ must be one of the endpoints.
    Thus, 
    \begin{align*}
          c[t, X, \P, \iPi] = \min_{(\P_1, \P_2)\in \mathbf{P}, 0 \leq \mathit{\rho} \leq \iPi - |X|} \left\{c[t_1, X, \P_1, \mathit{\rho}] + c[t_2, X, \P_2, n - \mathit{\rho} - |X|]\right\},
    \end{align*}
    where $\mathbf{P}$ contains all pairs of partitions $\P_1$ and $\P_2$ such that $\P$ is a proper walk merge of $\P_1$ and $\P_2$. As the profit of vertices in $X$ is counted in both solutions corresponding to $\P_1$ and $\P_2$, $|X|$ must be subtracted. 
    Furthermore, the resulting budget must be distributed between $t_1$ and $t_2$.

    This completes the computation of $c$. As every bag has size at most $k + 2$ and profit $0 \leq \Pi \leq n$, the number of states per node is at most $n \cdot 2^{k+2} \cdot (k + 2)^{k + 2} = n \cdot k^{\mathcal{O}(k)}$.
    Since $T$ has $\mathcal{O}(n)$ nodes, the total running time of the algorithm is $\mathcal{O}(n^2) \cdot k^{\mathcal{O}(k)}$.
\end{proof}

While we mainly need result of \Cref{thm:main} to give a $(1+\eps)$-approximation algorithm for the general \OP on graphs with bounded tree-width, it also means that the \OP with unit vertex profits is fixed-parameter tractable for the parameter tree-width, since a tree decomposition of a fixed width can be computed in linear time~\cite{treewidth}. Note that this result has been obtained simultaneously to this work in \cite{ren2024approximation}, where the authors consider the P2P-Orienteering problem. This problem is very similar to the \OP as we define it, but includes start and end point of the required walk, instead of only a start point, and only considers unit profits. They show that the P2P-Orienteering problem is fixed-parameter-tractable on graphs with bounded tree-width by giving a different dynamic program than ours. 


Using similar techniques as in~\cite{chekuri2012improved, GAVALAS2015313, korula2010approximation}, \Cref{thm:main} can be extended to an approximation for graphs with bounded \tw and arbitrary profits.

\begin{theorem}
    There exists a $(1 + \eps)$-approximation algorithm for the \OP on graphs with bounded \tw $k$ that takes $\mathcal{O}(n^6) \cdot k^{\mathcal{O}(k)}$ time.
\end{theorem}
\begin{proof}
We modify a given instance $((V, E), \pi, c, s)$ s.t.~all vertices have unit profit and use Theorem \ref{thm:main} to solve this new instance. 
First guess the highest vertex profit $\pi_{max}$ of an optimal walk and remove all vertices with a higher profit.
Create a new instance $(V', E')$ by copying $(V, E)$ and adding, for each 
$v \in V$, $\lfloor n^2 \pi(v)
/\pi_{max}\rfloor$ additional vertices connected to $v$ with an edge of cost $0$. 
Let $\pi'(v) = \lfloor n^2 \pi(v)/\pi_{max}\rfloor + 1$, i.e.~when a walk includes $v$ we assume that it also includes all vertices that are connected with $0$ cost edges.
Now consider a feasible walk $W$ in $(V,E)$ of distinct vertices $v_1, \cdots, v_\ell$ which has profit $\pi(W) = \sum^\ell_{i=1}\pi(v_i)$. The same walk $W'$ in $(V', E')$ (including our newly added vertices) collects $\pi'(W') = \sum^\ell_{i=1}\pi'(v_i) > n^2 \pi(W)/\pi_{max}$.
Let $\OPT'$ and $\OPT$ be the scores of the optimal walks in the modified and original instance, respectively.
Then, $\OPT' > ({n^2}/{\pi_{max}})\OPT$, however, $\pi'(W) \leq \sum^\ell_{i = 1}{n^2 \pi(v_i)}/{\pi_{max}} + 1$.
Hence, $({n^2}/{\pi_{max}}) \pi(W) \geq \pi'(W') - n \geq \pi'(W') - n \cdot {\OPT}/{\pi_{max}}$, and thus we obtain 
\[\pi(W) \geq \frac{\pi_{max}}{n^2}\pi'(W') - \frac{1}{n} \OPT.\]

As the aforementioned transformations do not increase the \tw and the result of Theorem \ref{thm:main} can be applied to obtain solution $W'$ with running time $\mathcal{O}(n^6) \cdot k^{\mathcal{O}(k)}$, solution $W'$ has profit $\pi(W') \geq (1 - \frac{1}{n})\OPT$. Hence, when $\frac{1}{n} < \eps$, we can brute force a solution, otherwise the above reduction can be followed.
\end{proof}



Note that we cannot expect to find an exact polynomial-time algorithm for the \OP on graphs with bounded tree-width (unless P=NP), since it is already NP-hard for $k=1$, i.e.~on trees.

\begin{proposition} \label{prop-op-np-trees}
    The \OP is NP-hard even on trees.
\end{proposition}

\begin{proof}
    We prove NP-hardness by a reduction from the \textsc{Knapsack} problem. Recall the decision variant of the \textsc{Knapsack} problem: Given a set of $n$ items with sizes $s_1, s_2, \cdots, s_n$, values $p_1, p_2, \cdots, p_n$, a capacity $C$, and a target value $P$, does a subset $S \subseteq \{1, 2, \cdots n\}$ exist, such that $\sum_{i \in S} s_i \leq B$ and $\sum_{i \in S} p_i \geq P$?
      
    Given such an instance, construct an \OP instance $((V, E),$ $\pi, c, s, B)$. Let $s$ be the start vertex with $\pi(s) = 0$. For each item $i$, add a vertex $v_i$ with profit $\pi(v_i) = p_i$, add an edge $e_i = (s, v_i)$ to $E$ with cost $c(e_i) = s_i/2$ and set $B = C$. The graph of our constructed instance is a tree.
    A tour that collects maximum profit on this instance would immediately give an optimal solution to the corresponding knapsack instance.
\end{proof}

\section{Conclusion and Open Problems}

In this paper, we explore various versions of the \OP 
on tree-like and path-like graph structures. 
The best known approximation factor for Orienteering on general undirected graphs is $(2+\eps)$~\cite{chekuri2012improved}. For graphs with bounded \tw we achieve a $(1+\eps)$-approximation. Can this approximation factor be extended to a wider class of graphs, in particular planar graphs? For closely related problems like the Prize-Collecting Stroll Problem the $(1+\eps)$-approximation algorithm on planar graphs builds upon an algorithm for bounded \tw~\cite{bateni2011prize}.

For the \OP with time windows we focus on the case of walks on directed cycles and one time window per profit. This is motivated by the fact that while the problem is easy to solve on directed paths, it is already NP-hard for undirected paths. For directed cycles we successfully give efficient algorithms for several fundamental cases, including the case that all profits need to be collected, but leave the complexity of this problem open in general. We do give a $(1+\eps)$-approximation, and it would be interesting to extend this result to a wider class of graphs,  even at the cost of a larger approximation factor, since for general graphs not even an $\mathcal{O}(\log n)$ approximation is known~\cite{Korula16survey}.



Finally,
in applications of the \OP, frequently the Edge Orienteering Problem, where profits are given for travelling over edges instead of vertices, is considered~\cite{GAVALAS2015313, verbeeck2014extension, lu2015arc}, also with time windows~\cite{lu2017scenic}.
All our results could be extended to Edge Orienteering and its versions with time windows and dynamic graphs.




\bibliography{refs}

\end{document}